\documentclass{article}

\usepackage{fullpage}

\usepackage{hyperref}
\usepackage{amsmath}
\usepackage{amssymb}
\usepackage{amsthm}
\usepackage{graphics}
\usepackage{graphicx}
\graphicspath{ {fig/} }
\usepackage{subfigure}
\usepackage{pifont}
\usepackage{algorithm2e}

\usepackage{enumitem}

\usepackage{float} 
\newfloat{algorithm}{thp}{lop}
\floatname{algorithm}{Algorithm}

\def\R{{\mathbb R}} 
\def\N{{\mathbb N}}

\def\bp{\noindent {\it Proof.}\ }
\def\ep{\hfill $\Box$} 

\newtheorem{theo}{Theorem}
\newtheorem*{theo:stability}{Theorem \ref{theo:stability}}

\newtheorem{lemma}{Lemma}
\newtheorem{coro}{Corollary}

\usepackage{authblk}

\begin{document}

\title{Balanced Fair Resource Sharing in Computer Clusters\footnote{\copyright\ 2017. This manuscript version is made available under the CC-BY-NC-ND 4.0 license \url{http://creativecommons.org/licenses/by-nc-nd/4.0/}.} \footnote{DOI: 10.1016/j.peva.2017.08.006}}

\author[1]{Thomas Bonald}
\author[2,1]{C\'eline Comte\thanks{The authors are members of LINCS, see \href{http://www.lincs.fr}{http://www.lincs.fr}.}\thanks{Emails: thomas.bonald@telecom-paristech.fr, celine.comte@nokia.com}}
\affil[1]{T\'el\'ecom ParisTech, Universit\'e Paris-Saclay, France}
\affil[2]{Nokia Bell Labs, France}

\date{\today}
\maketitle

\begin{abstract}
We represent a computer cluster as a multi-server queue 
 with some arbitrary  {graph of compatibilities} between jobs and servers.
 Each server processes its jobs sequentially in FCFS order.
 The service rate of a job at any given time is the sum of the service rates of all servers processing this job.
 We show that the corresponding queue is quasi-reversible and use this property to design a scheduling algorithm 
 achieving {balanced fair} sharing of the computing resources. 
\end{abstract}

\section{Introduction}
\label{sec:intro}

Load balancing is a  critical component of large-scale computer clusters.
The flow of requests must be directed to the servers under various  constraints like data availability, state of the servers and  service level agreements.
In this paper, we represent these  constraints by an arbitrary  {graph of compatibilities} between jobs and servers.
The computer cluster can then be viewed as a multi-server queue
where jobs are allocated to servers according to this graph.
We assume that each server processes its jobs sequentially in FCFS order.
The service rate of a job at any given time is the sum of the service rates of all servers processing this job,
which means that resource pooling does not induce any processing overhead.
We prove  that, for Poisson job arrivals and exponential job sizes,  this multi-server queue is quasi-reversible \cite{kelly}.
Exploiting this property, we 
design a novel scheduling algorithm achieving \emph{balanced fair} sharing of the computing resources. This makes the stationary distribution of the system state insensitive to the job size distribution beyond the mean, a practically interesting property leading to simple and robust engineering rules. 

Balanced fairness was introduced in the context of data networks as the most efficient resource allocation having the insensitivity property, 
allowing the service provider to develop dimensioning rules based on average traffic only, and not on detailed traffic characteristics \cite{BP03-1}.
Formally, it is the only allocation such that the underlying Markov process is reversible and at least one resource is saturated in each state.
Balanced fairness has later been used to evaluate the performance  of content-distribution networks \cite{SV15,SV16}.
However, no scheduling algorithm has been proved so far to achieve this allocation, except in some specific cases where it coincides with proportional fairness \cite{KMW09,W12}. 
To the best of our knowledge, our scheduling algorithm is the first practical implementation of balanced fairness,
just like the round-robin scheduling algorithm is a well-known practical implementation of the ideal processor-sharing (PS) service discipline.

Multi-server queues with specialized servers have already been considered in \cite{C09,V12,A14,T15}
but these models assume that each job can be processed by only one server at a time.
Our model is closer to the multi-server queue with redundant requests introduced by Gardner et al.~\cite{G15,G16},
where the class of a job defines the set of servers on which it is replicated.
When several replicas of the same job are in service simultaneously on different servers,
their service times are independent
and the first instance to be completed stops the others.
It is easy to see that, under the assumption of exponential service times,
the two models are in fact equivalent.
In both cases, the FCFS policy makes the system very sensitive to the job size distribution,
so that the actual performance may vary significantly when the job sizes are not exponentially distributed with the same unit mean.
Our objective in this paper is precisely to relax this assumption
by designing a scheduling policy which makes the system insensitive to the job size distribution.

It turns out that our model belongs to the family of Order Independent (OI) queues \cite{B96,K11}.
As observed in \cite{K11},
OI queues generalize a number of queueing systems like BCMP networks under the FCFS or PS service discipline \cite{BCMP75},
multiserver stations with concurrent classes of customers (MSCCC)
and multiserver stations with hierarchical concurrency constraints (MSHCC)
\cite{MSCCC,MSHCC}.
OI queues are known to be quasi-reversible \cite{kelly}.
In particular, the state of the queue has an explicit stationary distribution
under the usual assumptions of Poisson arrivals and exponential service times.
Moreover, the stationary distribution remains explicit in the presence of random routing,
where jobs can leave or re-enter the queue upon service completion.

The first contribution of this paper is a scheduling algorithm
which exploits this last property to mitigate the sensitivity to the job size distribution.
Just like round-robin scheduling which implements the PS service discipline in the single-server case,
our mechanism enforces insensitivity by interrupting the jobs frequently and moving them to the end of the queue.
Routing is thus reinterpreted in terms of job interruptions and resumptions.
The queue state is updated in the course of the job shiftings
and the exponentially distributed sizes with unit mean in the multi-server queue
now represent small fragments of the jobs.
When the interruptions are frequent,
each job tends to go back and forth in the queue
and its average service rate is mainly determined
by the number of jobs of each class which are present at the same time.

This last observation motivates us to adopt a higher viewpoint.
Specifically, we aggregate the state of the multi-server queue to only retain the number of jobs of each class,
but not their arrival order.
This aggregate state turns out to be an appropriate level of granularity to analyze the behavior of the queue.
Its stationary measure is exactly that of a Whittle network \cite{S99}
containing as many PS queues as there are classes in the original multi-server queue.
This leads us to our second contribution: a new theoretical understanding of the multi-server queue.
Using the state aggregation,
we show in Theorem \ref{theo:stability} that the queue is stable under any vector of acceptable arrival rates.
In practice, it suggests that our algorithm will tend to stabilize the system whenever possible.
Our second theoretical result, stated in Theorem \ref{theo:average},
concerns the service rate received on average by each job as its position in the queue evolves.
We show that the average per-class service rates when the number of jobs of each class is given
are exactly those obtained by applying balanced fairness.

In addition to help us to understand the behavior of our algorithm,
this equivalence with balanced fairness allows us to derive
explicit expressions for the performance metrics of the multi-server queue with an arbitrary graph of compatibilities.
Indeed, the insensitivity property satisfied by balanced fairness 
was used for instance in \cite{BV04,SV15,SV16} to
obtain simple and explicit recursion formulas for the performance metrics.
Thanks to the aggregation we propose,
these formulas can be applied as they are in the multi-server queue.
They predict the exact performance of our algorithm when the job sizes are exponentially distributed.
For an arbitrary job size distribution, we show by simulation that
the system becomes approximately insensitive when the number of interruptions per job increases,
so that the performance tends to that obtained under balanced fairness.
We further observe that only a few interruptions per job actually suffice to reach approximate insensitivity.

The rest of the paper is organized as follows.
In Section \ref{sec:model}, we introduce the model
and give the stability condition after recalling results on OI queues.
In Section \ref{sec:whittle}, it is shown that  the  resource allocation is balanced fairness in the presence of reentrant jobs.  
This result is used in Section \ref{sec:algo} to design our  scheduling algorithm. 
Some numerical results are presented in Section \ref{sec:num}.
Section \ref{sec:ccl} concludes the paper.

\section{A multi-server queue}
\label{sec:model}

We consider a multi-server queue with $N$ job classes and $S$ servers.
The class of a job may identify a client of the data center or a type of service;
it defines the set of servers that can process this job.
For each $i = 1,\ldots,N$, class-$i$ jobs enter the queue according to an independent Poisson process of intensity $\lambda_i$.
The job sizes are independent, exponentially distributed with mean $1$.
We assume for now that each job leaves the queue immediately after service completion.

For each $i = 1,\ldots,N$, we denote by ${\cal S}_i \subset \{1,\ldots,S\}$
the set of servers that can process class-$i$ jobs.
Equivalently, these constraints can be represented as a bipartite graph of compatibilities
between the $N$ job classes and the $S$ servers,
where there is an edge between class $i$ and server $s$ if and only if $s \in {\cal S}_i$.
Each job can be served in parallel by multiple servers
and each server processes the job sequentially in FCFS order.
Hence, when there are several servers available for a job at its arrival, all these servers process this job.
When the service of a job is complete, all the servers that were processing it
are reallocated to the next job they can serve in the queue.
There is no service preemption, so that at most one job of each class can be served at any given time.

We describe the evolution of the sequence of jobs in the queue, ordered by their arrival times.
Thus the queue state is some sequence $c = (c_1,\ldots,c_n)$ of length $n$,
where $n$ is the number of jobs in the queue
and $c_k \in \{1,\ldots,N\}$ is the class of job in position $k$, for each $k = 1,\ldots,n$,
starting from the head of the queue.
$\emptyset$ denotes the empty state, with $n = 0$.

When a job is in service on several servers,
its service rate is the sum of the capacities of the servers that are processing it.
Denoting by $\mu_s > 0$ the capacity of server $s$ for each $s = 1,\ldots,S$,
the total service rate in any state $c$ is thus given by
$$
\mu(c) = \sum_{s \in \bigcup_{k=1}^n {\cal S}_{c_k}} \mu_s.
$$
For each $k = 1,\ldots,n$, the job in position $k$ receives service at rate
$$
\mu(c_1,\ldots,c_k) - \mu(c_1,\ldots,c_{k-1})
= \sum_{s \in {\cal S}_{c_k} \setminus \bigcup_{\ell=1}^{k-1} {\cal S}_{c_\ell}} \mu_s.
$$
Observe that the total service rate in state $c$ only depends on the set $A(c) = \{c_k : k = 1,\ldots,n\}$ of active classes in state $c$.
Hence, for each $A \subset \{1,\ldots,N\}$, we can denote by $\mu(A)$ the service rate in any state $c$ whose set of active classes is $A$.
This is a submodular function, as a weighted cover set function \cite{E03,N88}.

\paragraph{Order Independent queues}

This multi-server queue turns out to be a special case of Order Independent (OI) queues.
These were introduced by Berezner and Krzesinski \cite{B96,K11} as a new class of multi-class quasi-reversible queues.
The description of an OI queue is the same as for the multi-server queue
except that the total service rate $\mu$ can be any function of the queue state $c$ which satisfies the following properties:
\begin{itemize}
  \item Monotonicity: $\mu(c_1,\ldots,c_n) \le \mu(c_1,\ldots,c_n,i)$ for any state $c$ and class $i$,
  \item Order-independence: $\mu(c_1,\ldots,c_n) = \mu\left( c_{\sigma(1)}, \ldots, c_{\sigma(n)} \right)$
    for any state $c$ and permutation $\sigma$ of $1,\ldots,n$.
\end{itemize}
Additionally, it is assumed that $\mu(\emptyset) = 0$ and $\mu(c) > 0$ for all $c \neq \emptyset$.
The total service rate $\mu(c)$ is allocated to jobs in the order of their arrival
in the sense that the job in position $k$ receives service at rate
$\mu(c_1,\ldots,c_k) - \mu(c_1,\ldots,c_{k-1})$.
In particular, the service received by a job does not depend on the jobs arrived later in the queue.
One can easily verify that the service rate of our   multi-server queue satisfies these properties.

\paragraph{Stationary measure}

The queue state $c$ defines a Markov process on $\{1,\ldots,N\}^*$.
Since the multi-server queue is a special case of OI queues,
it follows from \cite[Theorem 2.2]{K11} that this queue is quasi-reversible, with stationary measure
\begin{equation}
  \label{eq:stationary}
  \forall c \in \{1,\ldots,N\}^*, \quad
  \pi(c) = \pi(\emptyset) \prod_{k=1}^n \frac{\lambda_{c_k}}{\mu(A(c_1,\ldots,c_k))}.
\end{equation}
This formula was also derived in \cite[Theorem 1]{G15} for multi-server queues with redundant requests.
However, the observation that the multi-server queue is quasi-reversible is critical because it allows us to add random routing between job classes \cite{kelly}.
As we will see in Sections \ref{sec:whittle}  and \ref{sec:algo},
this result plays a key role in the design of our algorithm.

\paragraph{Aggregate state}

As in \cite{K11}, we consider the number of jobs of each class in the queue, independently of their arrival order.
We denote by $x = (x_1,\ldots,x_N)$ the corresponding aggregate state, where $x_i$ is the number of class-$i$ jobs in the queue.
This defines a stochastic process on $\N^N$, which is not a Markov process in general.
We refer to the stationary measure of the aggregate state $x$ as
\begin{equation}
  \label{eq:defpix}
  \bar\pi(x) = \sum_{c:|c|=x} \pi(c),
\end{equation}
where $|c| \in \N^N$ denotes the vector of the numbers of jobs of each class in state $c$.
We also denote the set of active classes in any state $x$ by $A(x) = \{i: x_i > 0\}$.

It was proved in \cite{K11} that the stationary measure of the aggregate state is given by
\begin{equation}
  \label{eq:recpix}
  \bar\pi(x) = \bar\pi(0) \Phi(x) \prod_{i=1}^N \lambda_i^{x_i},
\end{equation}
where the function $\Phi$ satisfies the recursion $\Phi(0) = 1$ and, for each $x \neq 0$,
\begin{equation}
  \label{eq:recPhi}
  \Phi(x) = \frac1{\mu(A(x))} \sum_{i \in A(x)} \Phi(x-e_i),
\end{equation}
$e_i$ being the $N$-dimensional vector with $1$ in component $i$ and $0$ elsewhere, for any $i = 1,\ldots,N$.

\paragraph{Stability condition}

The following key result is proved in the appendix.

\begin{theo}
  \label{theo:stability}
  The multi-server queue is stable,
  in the sense that the underlying Markov process is ergodic,
  if and only if
  \begin{equation}
    \label{eq:stability}
    \forall A \subset\{1,\ldots,N\}, A \neq \emptyset, \quad
    \sum_{i \in A} \lambda_i < \mu(A).
  \end{equation}
\end{theo}

In the rest of the paper, we assume that this condition is satisfied and we denote by $\pi$ the stationary distribution of the queue state.

\section{Average resource allocation}
\label{sec:whittle}

\paragraph{Re-entrant jobs}

Since the multi-server queue is quasi-reversible,
the stationary distribution of the queue state is not modified by the addition of routing between classes
as long as the effective arrival rates remain constant \cite{kelly}.
Assume for instance that each job leaves the queue with probability $p$ and re-enters as a job of the same class with probability $1-p$, for some $p \in (0,1]$.
The external arrival rate of class-$i$ jobs is taken equal to $\lambda_i p$ so that the effective arrival rate of class-$i$ jobs remains equal to $\lambda_i$.
The stationary distribution of the queue state $c$ is still given by \eqref{eq:stationary}, independently of $p$.
Each job re-enters the queue $1/p$ times on average, which tends to infinity when $p \to 0$.

In the limit, it is not relevant to consider the instantaneous service rate of each job depending on its position in the queue;
the metric of importance is the service rate received {\it on average} by each job
when the number of jobs of each class in the queue is given, corresponding to the aggregate state $x$.
The objective of this section is precisely to gain insights into the steady-state behavior of the multi-server queue viewed through its aggregate state.

\paragraph{Whittle network}

As we will see in Theorem \ref{theo:average} below,
the stationary distribution \eqref{eq:recpix} of the aggregate state $x$ of the multi-server queue
is that of the state of a Whittle network \cite{S99} of $N$ queues.

A Whittle network of $N$ queues
is a network of $N$ processor-sharing queues with state-dependent service rates.
The network state is described by the vector $x = (x_1,\ldots,x_N)$ giving the number of jobs at each queue.
The key feature of a Whittle network is that the relative variations of the service rates $\phi_1,\ldots,\phi_N$ of the queues
are constrained by the following balance property:
\begin{align}
  \label{eq:bal}
  \forall x \in \N^N, \quad
  \forall i,j: x_i > 0, x_j > 0, \quad
  \phi_i(x) \phi_j(x-e_i) = \phi_i(x-e_j) \phi_j(x).
\end{align}
This balance property is equivalent to the insensitivity property, i.e., the fact that the stationary distribution of the network state is 
independent of the job size distribution beyond  the mean
\cite{BP02}.

The service rates $\phi_1,\ldots,\phi_N$ satisfy the balance property \eqref{eq:bal}
if and only if there is a balance function $\Phi$ such that $\Phi(0) = 1$ and
\begin{align}
  \label{eq:Phi}
  \forall x \in \N^N, \quad
  \forall i = 1,\ldots,N, \quad
  \phi_i(x) = \begin{cases}
    \frac{\Phi(x-e_i)}{\Phi(x)} &\text{if } x_i > 0, \\
    0 & \text{otherwise}.
  \end{cases}
\end{align}
From this it is easy to show that the steady-state distribution $\bar\pi$ given by \eqref{eq:recpix}
with the function $\Phi$ given by \eqref{eq:Phi}
satisfies the local balance equations of the network.

Conversely, the balance function $\Phi$ uniquely defines
the service rates of the queues of a Whittle network through \eqref{eq:Phi}.
In particular, there exists a unique Whittle network of $N$ queues
with per-queue arrival rates $\lambda_1,\ldots,\lambda_N$
whose balance function is given by \eqref{eq:recPhi}.
The stationary distribution of this network state is
exactly the stationary distribution \eqref{eq:recpix} of the aggregate state of the multi-server queue.

The following two key results specify the relation between
the average per-class service rates in the multi-server queue
and the service rates of the queues in this equivalent Whittle network.

\begin{theo}
  \label{theo:average}
  The stationary distribution of the aggregate state of the multi-server queue
  is that of the state of a Whittle network of $N$ queues,
  with arrival rates $\lambda_1,\ldots,\lambda_N$ and state-dependent service rates $\phi_1,\ldots,\phi_N$ given by
  \begin{equation}
    \label{eq:average}
    \phi_i(x) = \sum_{c:|c| = x} \frac{\pi(c)}{\bar\pi(x)} \mu_i(c),
  \end{equation}
  where $\mu_i(c)$ is the service rate of the first class-$i$ job in state $c$ of the multi-server queue,
  for each $c \in \{1,\ldots,N\}^*$ and $i = 1,\ldots,N$.
\end{theo}

\begin{proof}
  As observed earlier, the stationary distribution \eqref{eq:recpix} is exactly the stationary distribution of the state $x$ of a Whittle network
  of $N$ queues with arrival rates $\lambda_1,\ldots,\lambda_N$ and service rates $\phi_1,\ldots,\phi_N$ given by \eqref{eq:bal},
  where $\Phi$ is the balance function given by \eqref{eq:recPhi}.
  We just need to verify that these service rates satisfy \eqref{eq:average}.

  Let $x \in \N^N$ and $i = 1,\ldots,N$ such that $x_i > 0$.
  We have
  \begin{align*}
    \phi_i(x)
    = \frac{\Phi(x-e_i)}{\Phi(x)}
    = \frac{\bar\pi(x-e_i) \lambda_i}{\bar\pi(x)}
    = \frac1{\bar\pi(x)} \sum_{c:|c|=x-e_i} \pi(c) \lambda_i.
  \end{align*}
  The quasi-reversibility of the multi-server queue ensures that the following partial balance equation is satisfied in any state $c$
  (see the proof of \cite[Theorem 2.2]{K11} for more details):
  \begin{align*}
    \pi(c) \lambda_i
    = \sum_{k=1}^{n+1} & \pi(c_1,\ldots,c_{k-1},i,c_k,\ldots,c_n) \\
                       & \times \left( \mu(A(c_1,\ldots,c_{k-1},i)) - \mu(A(c_1,\ldots,c_{k-1})) \right).
  \end{align*}
  Letting $n = x_1 + \ldots + x_N$, we deduce that
  \begin{align*}
    \sum_{c:|c|=x-e_i} \pi(c) \lambda_i
    =& \sum_{c:|c|=x-e_i}
    \sum_{k=1}^n \pi(c_1,\ldots,c_{k-1},i,c_k,\ldots,c_{n-1}) \\
    & \qquad \times \left( \mu(A(c_1,\ldots,c_{k-1},i)) - \mu(A(c_1,\ldots,c_{k-1})) \right), \\
    =& \sum_{c:|c|=x} \pi(c)
    \sum_{\substack{k=1 \\ c_k=i}}^n
    \left( \mu(A(c_1,\ldots,c_{k-1},c_k)) - \mu(A(c_1,\ldots,c_{k-1})) \right), \\
    =& \sum_{c:|c|=x} \pi(c) \mu_i(c).
  \end{align*}
  This equation remains valid for any state $x$ and class $i$ such that $x_i = 0$.
\end{proof}

\begin{coro}
  \label{coro:capa}
  For each $x \in \N^N$, the vector of service rates
  $\phi(x) = (\phi_1(x), \ldots, \phi_N(x))$
  belongs to the capacity set
  \begin{equation*}
    {\cal C} = \left\{
      \phi \in \R_+^N:~
      \forall A \subset \{1,\ldots,N\},~
      \sum_{i \in A} \phi_i \le \mu(A)
    \right\}
  \end{equation*}
  and satisfies
  \begin{equation*}
    \sum_{i \in A(x)} \phi_i(x) = \mu(A(x)).
  \end{equation*}
\end{coro}

\begin{proof}
  Let $x \in \N^N$.
  For all $A \subset \{1,\ldots,N\}$, we have by \eqref{eq:average},
  \begin{equation*}
    \sum_{i \in A} \phi_i(x)
    = \sum_{c:|c|=x} \frac{\pi(c)}{\bar\pi(x)} \sum_{i \in A} \mu_i(c)
    \le \sum_{c:|c|=x} \frac{\pi(c)}{\bar\pi(x)} \mu(A)
    = \mu(A).
  \end{equation*}
  For $A = A(x)$, we have
  \begin{equation*}
    \sum_{i \in A(x)} \mu_i(c)
    = \mu(A(x))
  \end{equation*}
  for each $c$ such that $|c|=x$, so that the above inequality is an equality.
\end{proof}

\paragraph{Balanced fairness}

By Theorem \ref{theo:average}, the average service rates in the multi-server queue satisfy the balance property \eqref{eq:bal}.
In view of Corollary \ref{coro:capa},
the resource allocation is also Pareto-efficient
in the sense that the server resources are always maximally consumed.
The unique resource allocation which satisfies these two properties is known as balanced fairness \cite{BP03-1}.

Going back to the motivating example with re-entrant jobs,
in the limit where $p \to 0$,
the external arrivals and departures become rare
and the jobs tend to re-enter the queue several times.
The detailed queue state evolves with these frequent job shifts,
while the aggregate state remains constant.
On average,
all jobs of class $i$ tend to be served at the same service rate in aggregate state $x$, with total service rate $\phi_i(x)$.
When the queue contains only one server,
it means that the capacity of this server is divided equally among all jobs in the queue, similarly to round-robin scheduling.
In general, this corresponds to the above Whittle network where each of the $N$ queues applies the processor-sharing service discipline.
Such a queueing system is known to have the insensitivity property described above.
This property will be exploited in the next section to design a scheduling algorithm in computer clusters based on re-entrant jobs after forced service interruptions.

\paragraph{Performance metrics}

Several works have focused on predicting the performance of systems under balanced fairness, see for instance \cite{BV04,SV15,SV16}.
Their results can be reused as they are to predict the performance of the multi-server queue.
Indeed, the above aggregation results show that any performance metric
which can be expressed in terms of the aggregate state in the multi-server queue
is actually equal to the corresponding metric in the equivalent Whittle network.
This is stated more formally in the following corollary, which follows from Theorem \ref{theo:average}.

\begin{coro}
  \label{coro:perf}
  Consider a function $f$ defined on $\{1,\ldots,N\}^*$ which is order-independent,
  in the sense that there exists a function $g$ defined on $\N^N$
  such that $f(c) = g(|c|)$ for all $c \in \{1,\ldots,N\}^*$.
  Then the expected value of $f$ applied to the state $c$ of the multi-server queue
  is equal to the expected value of $g$ applied to the state $x$ of the equivalent Whittle network.
\end{coro}

\begin{proof}
  Gathering the detailed queue states which correspond to the same aggregate state, we obtain directly
  \begin{align*}
    \sum_{c \in \{1,\ldots,N\}^*} \pi(c) f(c)
    = \sum_{c \in \{1,\ldots,N\}^*} \pi(c) g(|c|)
    = \sum_{x \in \N^N} \left( \sum_{c:|c|=x} \pi(c) \right) g(x)
    = \sum_{x \in \N^N} \bar\pi(x) g(x).
  \end{align*}
\end{proof}

Note that this result holds for any stationary measure $\pi$.
In particular, taking the measure $\pi$ such that $\pi(\emptyset) = 1$ in the multi-server queue
(that is, $\bar\pi(0) = 1$ in the equivalent Whittle network)
and for $f$ the constant function equal to $1$,
we obtain that the normalization constants in the multi-server queue and in the equivalent Whittle network are equal.

A metric of importance is the mean number of jobs of a given class in the multi-server queue,
from which we can deduce the mean delay (or equivalently the mean service rate)
perceived by the jobs of class $i$, for each $i = 1,\ldots,N$.
Coming back to the stationary distribution $\pi$
with the function $f$ which counts the number of jobs of a given class in the multi-server queue,
we deduce from Corollary \ref{coro:perf} that the mean number of class-$i$ jobs in the multi-server queue
is equal to the mean number of jobs at queue $i$ in the equivalent Whittle network,
for each $i = 1,\ldots,N$.

Hence, the recursive formulas of \cite[Theorem 4]{SV15} and \cite[Theorem 1]{SV16}
give directly the normalization constant of the stationary distribution,
as well as the mean number of jobs of each class in the multi-server queue.
The numerical results presented in Section \ref{sec:num} are based on  this observation.

\section{A scheduling algorithm for computer clusters}
\label{sec:algo}

We apply the previous results to the problem of resource sharing in computer clusters.
Consider a cluster of $S$ servers.
For all $s = 1,\ldots,S$, we denote by $C_s$ the service capacity of server $s$, in floating-point operations per second (flops).
Any incoming job consists of some random number of floating-point operations, referred to as the job size, and is assigned some set of servers.
This assignment, possibly random, may depend on the type of the job but not on the system state (e.g., the number of ongoing jobs).
It is fixed for the entire life of the job in the system.
The job can then be processed in parallel by any subset ${\cal S}$ of the servers in its assignment,
at rate $\sum_{s \in {\cal S}} C_s$. Job sizes are assumed i.i.d.~with mean  $\sigma$.

\paragraph{Balanced fairness}

We aim at sharing the service capacity of the cluster according to balanced fairness,
so that the stationary distribution of the number of jobs of each class  is independent of the job size distribution beyond the mean \cite{BP03-1}.
Applying the FCFS service discipline to each server is clearly not suitable.
For $S = 1$ for instance, the system reduces to a single-server FCFS queue,
which is known to be very sensitive to the job size distribution. For $S\ge 1$, the system corresponds to the multi-server queue described in Section \ref{sec:model}, with service rates   $\mu_s=C_s/\sigma$ for all $s=1,\ldots,S$,  provided  job sizes are i.i.d.~exponential with mean $\sigma$.

We apply the idea of re-entrant jobs mentioned in Section \ref{sec:whittle}.
Specifically, we  interrupt each service after some exponential time and  force  the corresponding job to re-enter  the queue as a new job of the same type, with the same server assignment, 
so that the service can be resumed later and the resources can be reallocated. 
Observe that, when job sizes are i.i.d.~exponential with mean $\sigma$, the stationary distribution of the aggregate state remains unchanged by the quasi-reversibility of the OI queue.
When the frequency of service interruptions increases, the resources tend to be shared fairly, in the sense of balanced fairness,
and the stationary distribution  becomes insensitive to the job size distribution beyond the mean.
For $S = 1$ for instance, the system tends to a single-server PS queue,
which is known to have the insensitivity property. For $S\ge 1$, the system tends to a Whittle network of PS queues, which is also known to have the insensitivity property \cite{BP02}.

\paragraph{Scheduling algorithm}

A single virtual queue is used to allocate servers to jobs.
Any incoming job is put at the end of the queue.
Each server $s$ interrupts the job in service, if any, after some exponential time with parameter $C_s/\theta$, for some $\theta > 0$.
Observe that $\theta$ can be interpreted as the mean number of floating-point operations before service interruption.
Any interrupted job releases {\it all} servers that process this job and is moved to the end of the queue as a new job.
The released resources are reallocated according to the same service discipline, accounting for the new order in the queue.
Note that the interrupted service may be resumed immediately or later, when some resources become available, depending on the state of the queue.

\begin{algorithm}[!tp]
  \begin{minipage}[t]{.5\textwidth}
    \begin{algorithm}[H]
      \SetKwBlock{Begin}{begin}{end}
      {\bf on} {\rm job arrival}\\
      \Begin{
        $i\gets$ job class\\
        enqueue job\\
        $n\gets n+1$\\
        \For{{\bf all} $s\in {\cal S}_i$}
        {
          \If{$a_s=0$}{$a_s\gets i$\\  $t_s\gets$ on}
        }
      }

      \

      {\bf on} {\rm job departure}\\
      \Begin{
        $i\gets$ job class\\
        $k\gets$ job position in the queue\\
        dequeue job\\
        $n\gets n-1$\\
        \For{{\bf all} $s\in {\cal S}_i$}
        {
          \If{$a_s=i$}
          {$a_s\gets 0$\\
            $t_s\gets$ off
          }
        }
        \While{$k \le n$}
        {
          \For{{\bf all} $s\in {\cal S}_{c_k}\cap {\cal S}_i$}
          {
            \If{$a_s=0$}
            {$a_s\gets c_k$\\
              $t_s\gets$ on
            }
          }
          $k\gets k +1$
        }
      }

    \end{algorithm}
  \end{minipage}
  \begin{minipage}[t]{.5\textwidth}
    \begin{algorithm}[H]
      \SetKwBlock{Begin}{begin}{end}
      {\bf on} {\rm timer expiration}\\
      \Begin{
        $s\gets$ server\\
        $i\gets a_s$ \\
        $k\gets$ job position in the queue\\
        interrupt job service  \\
        move job  to the end of the queue\\
        \For{{\bf all} $s\in {\cal S}_i$}
        {
          \If{$a_s=i$}
          {$a_s\gets 0$\\
            $t_s\gets$ off
          }
        }
        \While{$k \le n$}
        {
          \For{{\bf all} $s\in {\cal S}_{c_k}\cap {\cal S}_i$}
          {
            \If{$a_s=0$}
            {$a_s\gets c_k$\\
              $t_s\gets$ on
            }
          }
          $k\gets k +1$
        }
      }
    \end{algorithm}
  \end{minipage}
  \caption{\label{algo}Scheduling algorithm based on random service interruptions.}
\end{algorithm}

The pseudo-code of the algorithm is given in Algorithm \ref{algo},
where $a_s \in \{0,1,\ldots,N\}$ denotes the activity state of server $s$
($a_s = 0$ if server $s$ is idle and $a_s = i$ if server $s$ is processing a job of class $i$)
and $t_s \in \{\text{on},\text{off}\}$ indicates the state of the timer that triggers service interruption at server $s$
(when set on, the timer has an exponential distribution with parameter $C_s / \theta$).
The algorithm depends on a single parameter $\theta$, which determines the mean number of service interruptions per job. This should be compared to the mean job size 
$\sigma$. 
Specifically, the ratio $m=\sigma/\theta$ corresponds to the mean number of service interruptions per job. When $m\to \infty$, services are frequently interrupted and the corresponding resource allocation tends to balanced fairness, as mentioned in Section \ref{sec:whittle}, an allocation that  has the insensitivity property; when $m\to 0$, services are almost never interrupted and the service discipline is approximately FCFS per server, which is highly sensitive to the job size distribution.  
We shall see in the following section  that, for large systems with random assignment, setting $m=1$ is in fact sufficient to get approximate insensitivity, i.e.,  it is sufficient in practice to interrupt each job only {\it once} on average.

\section{Numerical results}
\label{sec:num}

In this section, we provide numerical results showing the performance of the algorithm described above.
We are specifically interested in evaluating the mean number $m$ of interruptions per job
which is sufficient in practice to obtain approximate insensitivity to the job size distribution.

\paragraph{Job size distribution}

As in Section \ref{sec:algo}, the job sizes are assumed i.i.d.
To test the sensitivity, we successively evaluate the performance of our algorithm under three job size distributions.

We first consider job sizes with a bimodal number of exponentially distributed phases.
More precisely, the size of any incoming job
is a sum of independent random variables which are exponentially distributed with mean $\varsigma$.
The number of these random variables follows a bimodal distribution:
it is equal to $n_1$ with probability $p_1$ and to $n_2$ with probability $p_2$, for some $p_1,p_2$ such that $p_1+p_2=1$.
We let $\varsigma=1/5, n_1 = 25$, $n_2 = 1$, $p_1 = 1/6$ and $p_2 = 5/6$.
The mean job size is given by $\sigma = (p_1 n_1 + p_2 n_2) \varsigma = 1$
while the standard deviation is approximately equal to $1.84$.

We consider a second alternative where the job size distribution is hyperexponential:
any incoming job has an exponential distribution with mean $\sigma_1$ with probability $p_1$
and an exponential distribution with mean $\sigma_2$ with probability $p_2$,
for some $p_1,p_2$ such that $p_1+p_2=1$.
We let $\sigma_1 = 5$, $\sigma_2 = 1/5$, $p_1 = 1/6$, $p_2 = 5/6$,
corresponding to a mean job size $\sigma = p_1 \sigma_1 + p_2 \sigma_2 = 1$
and standard deviation approximately equal to $2.05$.

Finally, we consider job sizes with a heavy-tailed number of exponential phases.
Like for the bimodal case,
the size of any incoming job is a sum of independent random variables which are exponentially distributed with mean $\varsigma$.
The number of these random variables follows a Zipf distribution with parameters $K \in \N$ and $\alpha > 0$:
for each $k = 1,\ldots,K$, the probability that there are $k$ terms in the sum is proportional to $1/k^\alpha$.
We let $\varsigma=1$, $K = 200$ and $\alpha = 2$.
The mean job size is then given by
$$
\sigma = \frac
{ \sum_{k=1}^K \frac1{k^{\alpha-1}} }
{ \sum_{k=1}^K \frac1{k^\alpha} }
\varsigma,
$$
approximately equal to $3.58$,
while the standard deviation is approximately equal to $10.61$.

\paragraph{Performance metrics}

We measure the performance in terms of mean service rate and mean delay. 
Let $\phi_i(x)$ be the service rate of class-$i$ jobs in state $x$, as defined by \eqref{eq:average}. 
The mean service rate of any class-$i$ job is then given by:
$$
\gamma_i=\frac{ \sum_{x} \bar\pi(x) \phi_i(x)}{\sum_{x} \bar\pi(x) x_i}.
$$
By conservation, we have:
$$
\gamma_i=\frac{\lambda_i \sigma}{\sum_{x} \bar\pi(x) x_i}.
$$
Observe that $\gamma_i$ cannot exceed the maximum service rate of class-$i$ jobs, given by
$\sum_{s\in {\cal S}_i}C_s$.
The mean delay $\delta_i$ of any class-$i$ job follows from the mean number of class-$i$ jobs by Little's law,
and is inversely proportional to the mean service rate:
\begin{equation}
  \label{eq:delta}
  \delta_i = \frac{\sum_x \bar\pi(x) x_i}{\lambda_i} = \frac\sigma{\gamma_i}.
\end{equation}

\paragraph{Performance evaluation}

There are $S$ servers and $N$ job classes.
Class-$i$ jobs arrive according to a Poisson process with intensity $\lambda_i$.
The mean number of interruptions per job is given by $m = \sigma / \theta$,
where $\sigma$ is the mean job size
and $\theta$ is the parameter of the algorithm used to set the random timers.
We compare the results for $m = 1$ and $m = 5$
with those obtained under FCFS policy (that is, without service interruption) and balanced fairness.

The performance metrics under balanced fairness will be given in closed form for the configurations considered below.
They give the performance of our algorithm and of FCFS policy when the job size distribution is exponential.
We resort to simulations to assess the performance under the three job size distributions listed earlier.
Each simulation point follows from the average of $100$ independent runs,
each corresponding to  $10^6$ jumps of the corresponding Markov process, after a warm-up period of $10^6$ points;
the corresponding $95\%$ confidence intervals are drawn in semitransparent on the figures.

\paragraph{Three servers}

We first consider a toy example with $S=3$ servers and 
$N=2$ job classes. Servers 1 and 2 are dedicated to job classes 1 and 2, respectively, while server 3 is shared by all jobs. 
In view of Theorem \ref{theo:stability}, the stability condition is:
$$
\lambda_1< \mu_1+\mu_3,\quad \lambda_2< \mu_2+\mu_3,\quad \lambda_1+\lambda_2< \mu_1+\mu_2+\mu_3.
$$
Define the corresponding loads: 
$$
\rho_1=\frac{\lambda_1}{\mu_1+\mu_3},\quad \rho_2=\frac{\lambda_2}{\mu_2+\mu_3},\quad \rho=\frac{\lambda_1+\lambda_2}{\mu_1+\mu_2+\mu_3}.
$$
Observing that the capacity set is that of a tree network \cite{BV04}, we deduce the mean service rates under balanced fairness:
\begin{align}
  \label{eq:tree1}
  \gamma_1 &= \left(
    \frac1{\mu (1-\rho)}
    + 
    \frac
    {\frac{\mu_2}{\mu_1+\mu_3} \frac{1-\rho_2}{1-\rho_1}}
    {\mu - (\mu_1+\mu_3) \rho_1 - (\mu_2+\mu_3) \rho_2 + \mu_3 \rho_1 \rho_2}
  \right)^{-1}, \\
  \label{eq:tree2}
  \gamma_2 &= \left(
    \frac1{\mu(1-\rho)}
    + \frac{
      \frac{\mu_1}{\mu_2+\mu_3}
      \frac{1-\rho_1}{1-\rho_2}
    }{\mu - (\mu_1+\mu_3) \rho_1 - (\mu_2+\mu_3) \rho_2 + \mu_3 \rho_1 \rho_2}
  \right)^{-1}
\end{align}
with $\mu=\mu_1+\mu_2+\mu_3$.
The mean delays follow by \eqref{eq:delta}.
Explicit formulas for the performance metrics under this assignment graph
were also derived in \cite{G15} in the context of multi-server queues with redundant requests.
Recall that these are the exact performance metrics when the job size distribution is exponential.
The results are shown in Figures \ref{fig:m} and \ref{fig:n} with respect to the load $\rho$, for $\lambda_1=\lambda_2$.
In Figure \ref{fig:m}, the system is symmetric and the maximum service rate is 2 for both classes.
In Figure \ref{fig:n}, the system is asymmetric:
class-1 jobs (in blue) can be served by servers 1,3 and thus have a maximum service rate of 2;
class-2 jobs (in red) can be served by server 3 only and thus have a maximum service rate of 1.

\begin{figure}[h]
  \centering
  \subfigure[Bimodal number of exponentially distributed phases]{\includegraphics{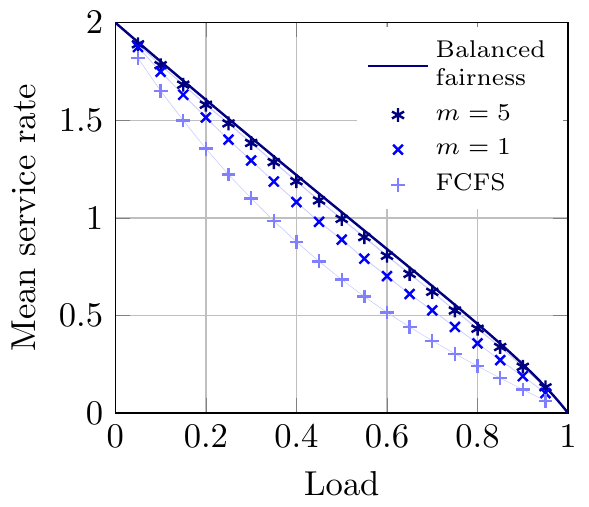} \hfill \includegraphics{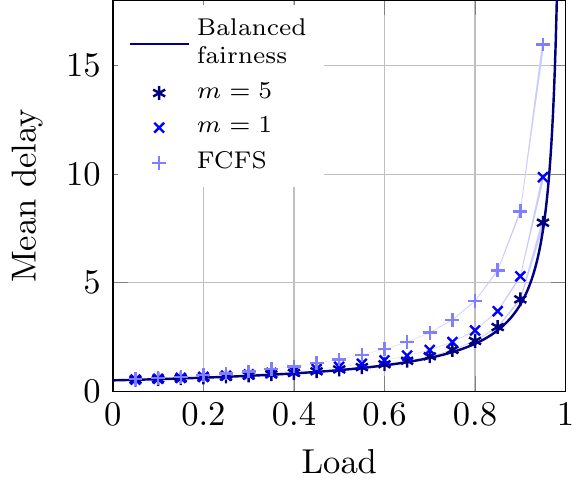}}
  \subfigure[Hyperexponential]{\includegraphics{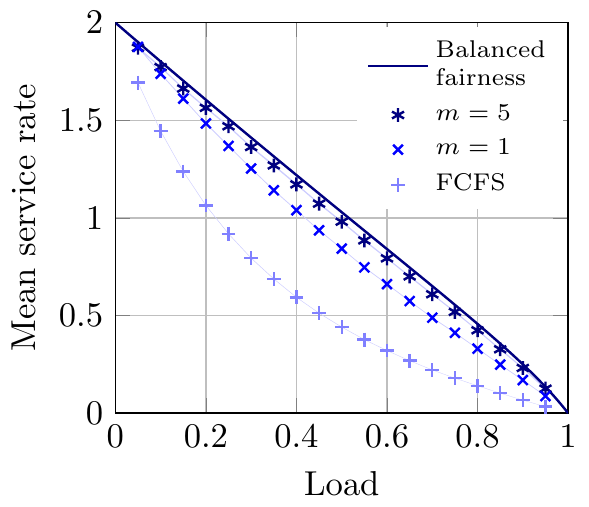} \hfill \includegraphics{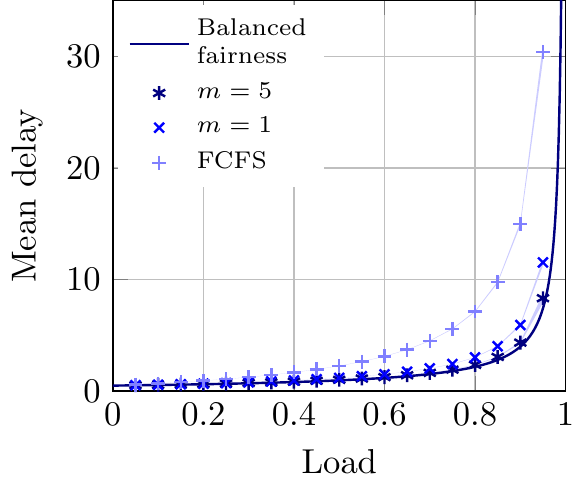}}
  \subfigure[Zipf number of exponentially distributed phases]{\includegraphics{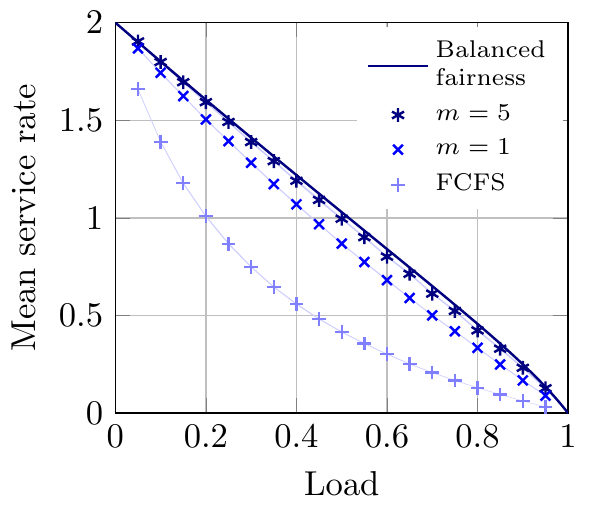} \hfill \includegraphics{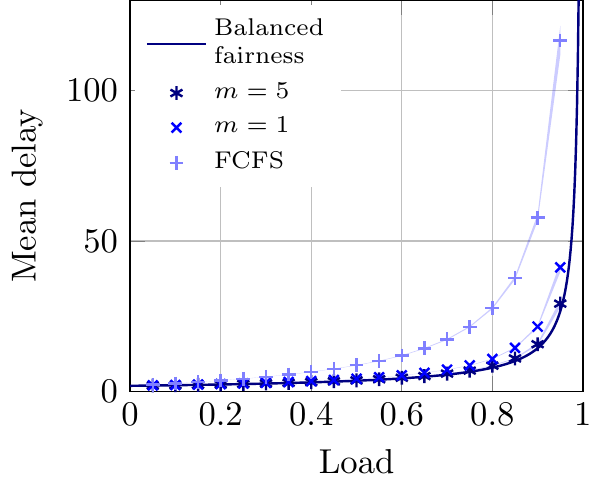}}
  \caption{Performance metrics for $N=2$ job classes sharing $S=3$ servers ($\mu_1 = \mu_2 = \mu_3 = 1$).}
  \label{fig:m}
\end{figure}

\begin{figure}[h]
  \centering
  \subfigure[Bimodal number of exponentially distributed phases]{\includegraphics{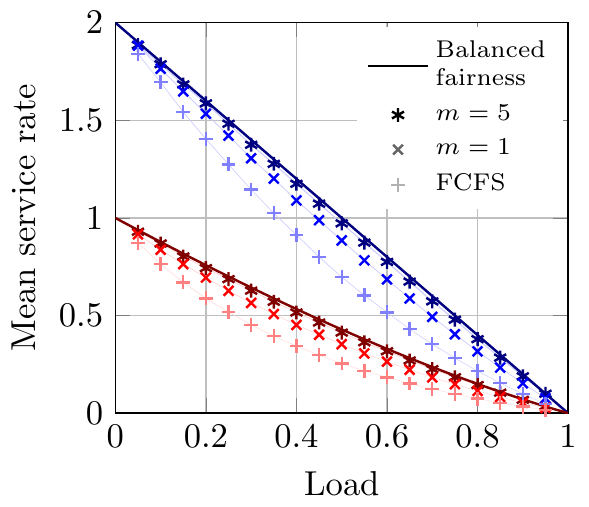} \hfill \includegraphics{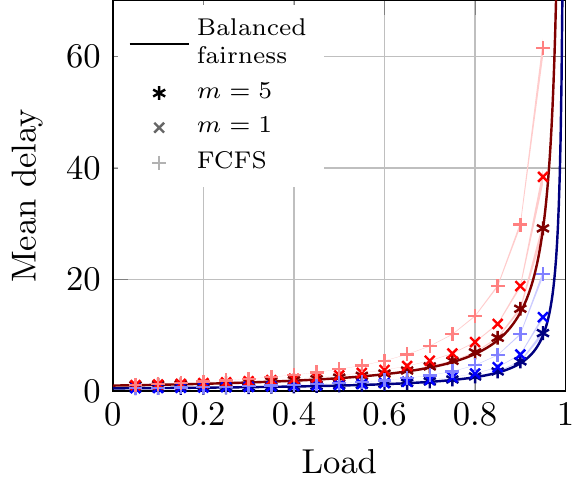}}
  \subfigure[Hyperexponential]{\includegraphics{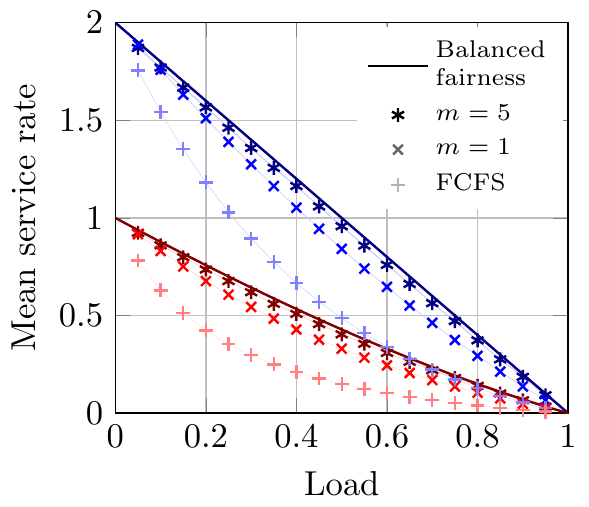} \hfill \includegraphics{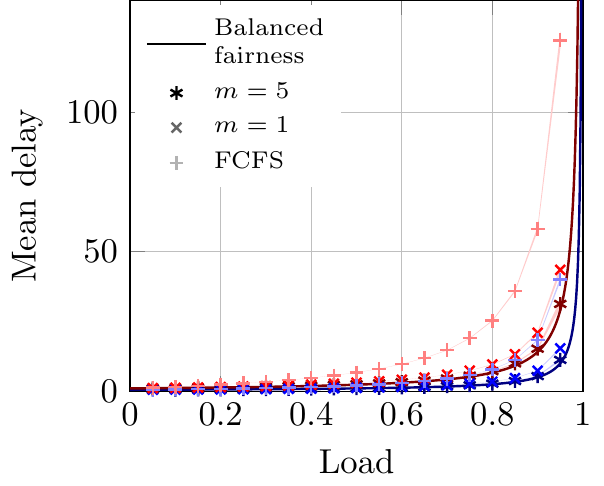}}
  \subfigure[Zipf number of exponentially distributed phases]{\includegraphics{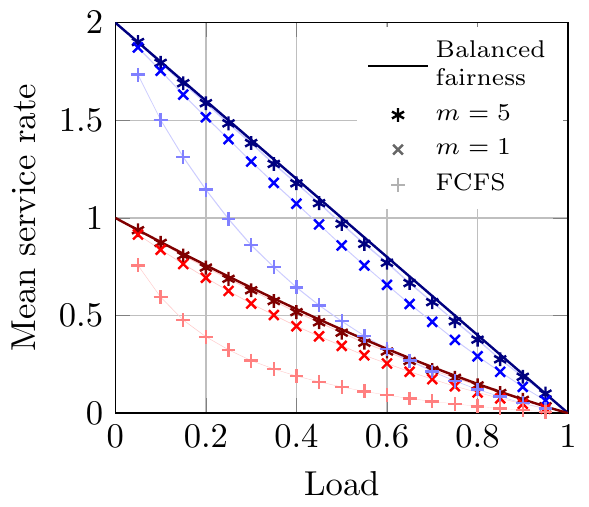} \hfill \includegraphics{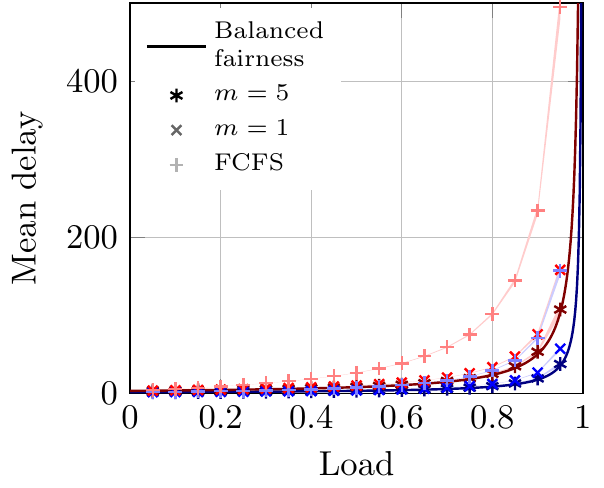}}
  \caption{Performance metrics for $N=2$ job classes sharing $S=2$ servers ($\mu_1 = \mu_3 = 1$, $\mu_2 = 0$).
  The performance of class-$1$ jobs, which have access to both servers, appears in blue on the figure
  (top plot for the mean service rate and bottom plot for the mean delay).
  The performance of class-$2$ jobs, which have access to only one server, appears in red on the figure
  (bottom plot for the mean service rate and top plot for the mean delay).
  }
  \label{fig:n}
\end{figure}

Applying our scheduling algorithm with only $m = 1$ (that is, 1 service interruption per job on average)
brings a significant improvement compared to FCFS policy.
For $m=5$, performance is very close to that of balanced fairness and approximately insensitive
(i.e., very close to that obtained for an exponential job size distribution)
even for job sizes with a Zipf number of exponential phases.

\paragraph{Large system with random assignment}

We now consider a large system of $S=100$ servers, each with unit service rate. Each incoming job is assigned $d$ 
servers chosen uniformly  at random, corresponding to  $N={S \choose d}$ job classes, as considered in \cite{G16}.
The mean service rate and the mean delay follow from
an explicit formula for the mean number of jobs in the queue derived in \cite{G16}.
The simulation results are obtained in the conditions described above.
The results for $d = 2$ and $d = 3$ are shown in Figures \ref{fig:random2} and \ref{fig:random3}, respectively.
We observe that performance is very close to that of balanced fairness, even for low values of $m$.
It is sufficient in practice to set $m=1$, corresponding to only {\it one} service interruption per job on average.

\begin{figure}[h]
  \centering
  \subfigure[Bimodal number of exponentially distributed phases]{\includegraphics{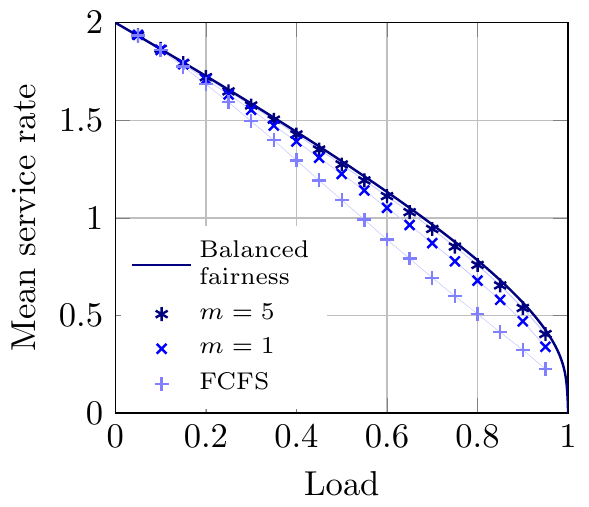} \hfill \includegraphics{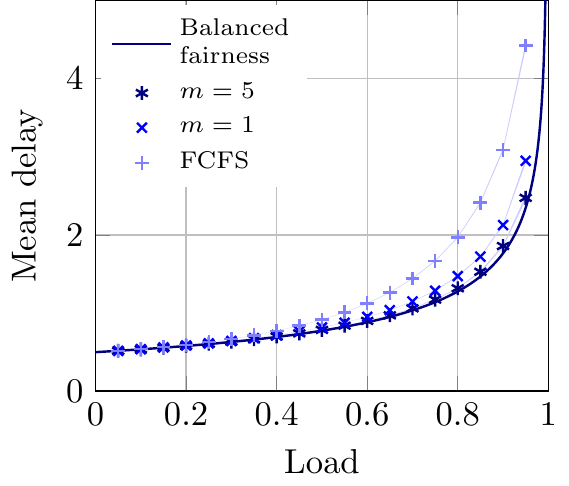}}
  \subfigure[Hyperexponential]{\includegraphics{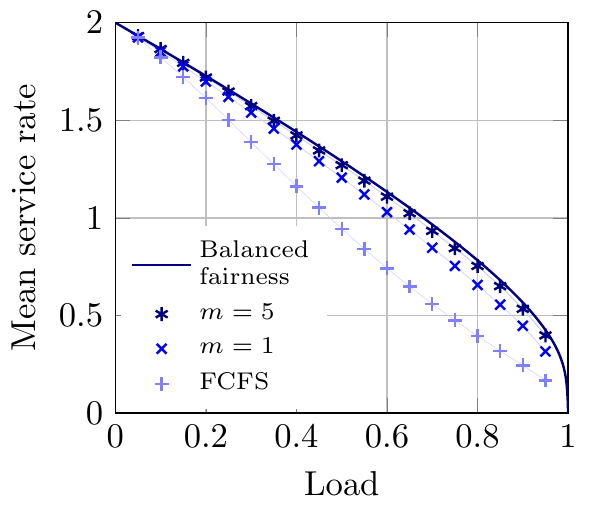} \hfill \includegraphics{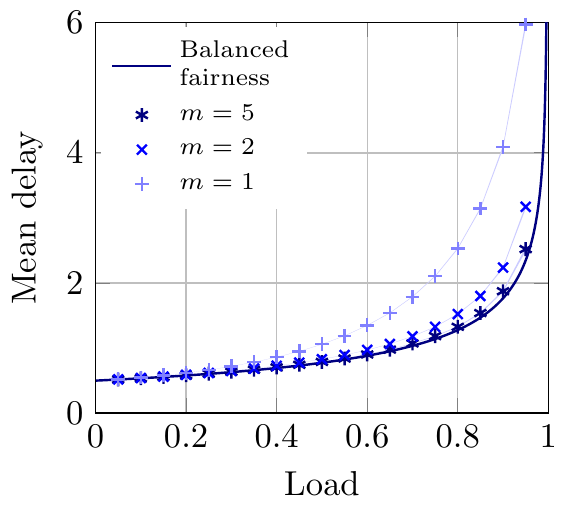}}
  \subfigure[Zipf number of exponentially distributed phases]{\includegraphics{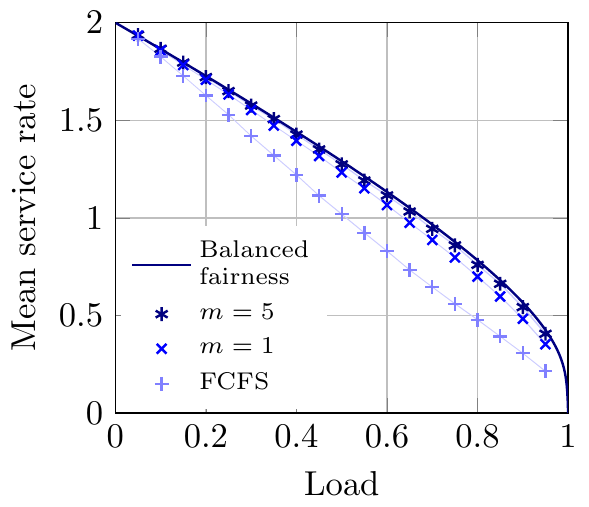} \hfill \includegraphics{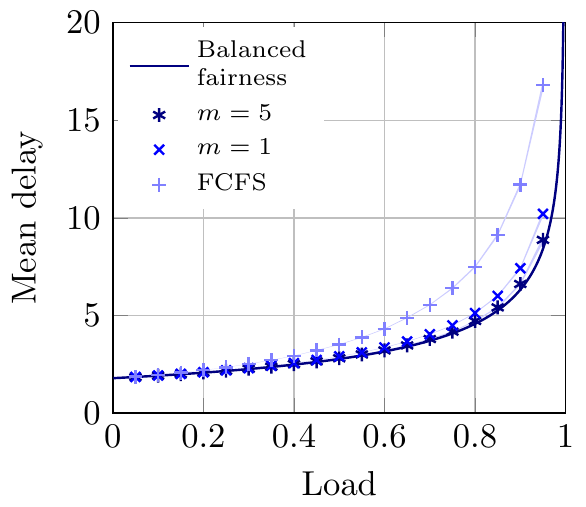}}
  \caption{Performance metrics for random assignment of $d = 2$ servers among $S = 100$.}
  \label{fig:random2}
\end{figure}

\begin{figure}[h]
  \centering
  \subfigure[Bimodal number of exponentially distributed phases]{\includegraphics{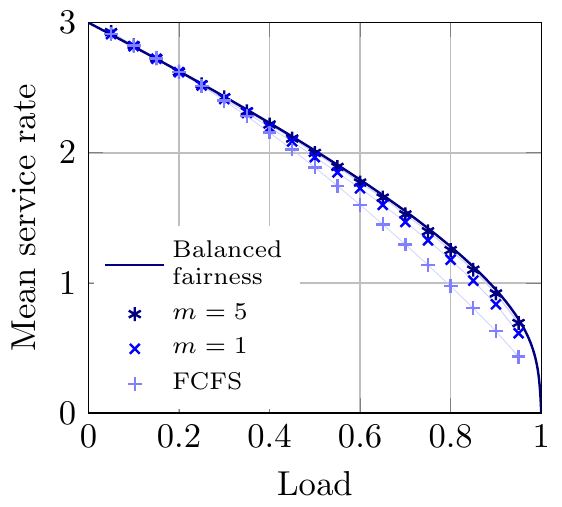} \hfill \includegraphics{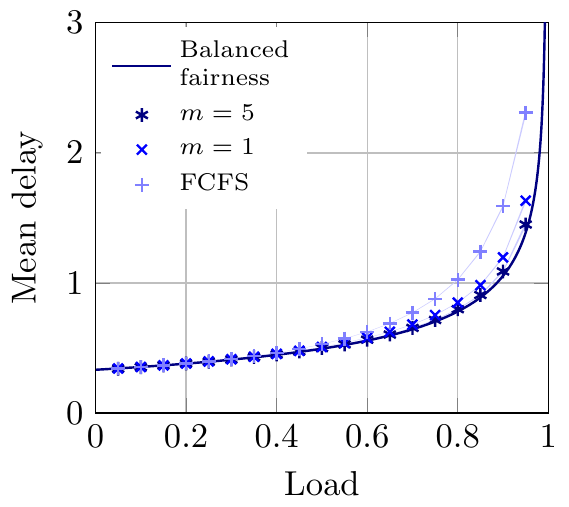}}
  \subfigure[Hyperexponential]{\includegraphics{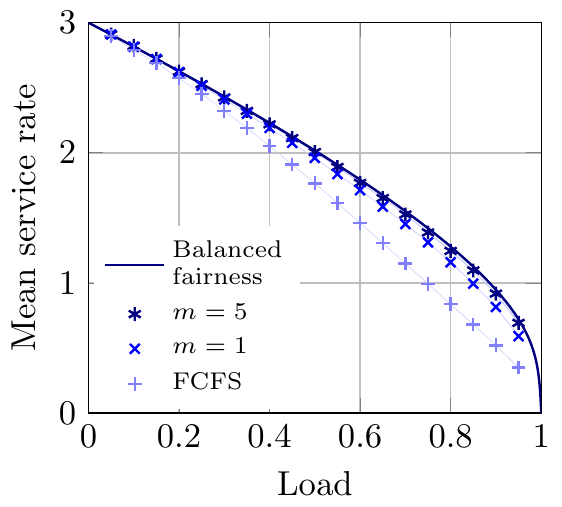} \hfill \includegraphics{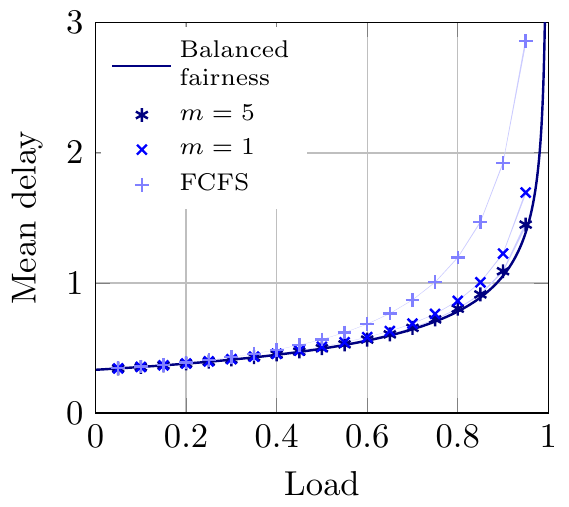}}
  \subfigure[Zipf number of exponentially distributed phases]{\includegraphics{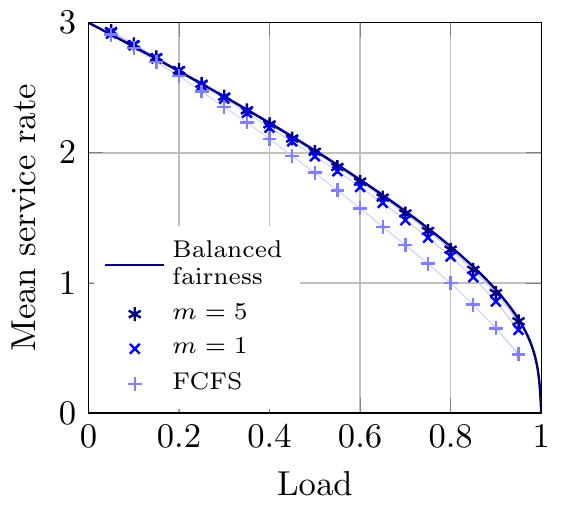} \hfill \includegraphics{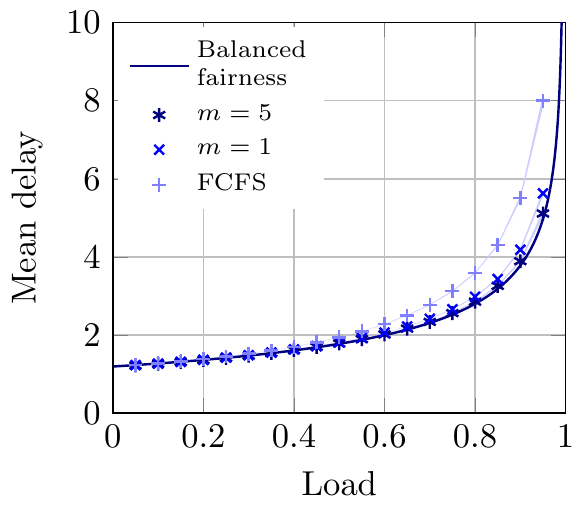}}
  \caption{Performance metrics for random assignment of $d = 3$ servers among $S = 100$.}
  \label{fig:random3}
\end{figure}

\section{Conclusion}
\label{sec:ccl}

We have introduced a new scheduling algorithm
to allocate the resources of a computer cluster according to balanced fairness.
This algorithm, which is based on service interruptions and resumptions,
can be viewed as an extension of round-robin scheduling algorithm in the context of resource pooling.
Its performance was studied by considering a new queueing model where jobs can be processed in parallel by several servers.
We have observed in particular that the aggregate state of the queue is that of a Whittle network,
and deduced the insensitivity property in the limit of an infinite number of service interruptions per job.
This has in turn allowed us to derive explicit expressions for the performance metrics
with an arbitrary graph of compatibilities.
The performance of the system when the number of service interruptions per job is finite was assessed by simulation.
We observed that only a few interruptions per job are sufficient in practice to obtain approximate insensitivity.

Our objectives for the future work are twofold.
First, we aim at refining our understanding of the system presented in this paper.
This notably involves assessing analytically the impact of the mean number of service interruptions
on the sensitivity of the resulting resource allocation.
We would also like to perform more simulations to compare the performance of our algorithm
with that of other existing scheduling policies,
regarding both the insensitivity to the job size distribution
and the efficiency of the resource utilization.

A second step would be to extend the current model and algorithm
and include practical constraints which are inherent to parallel computing.
We can notably mention the cost of coordination between servers, not only during the service but also upon service interruption.
Besides it would be interesting to consider alternative ways of enforcing frequent service interruptions,
which do not rely on exponentially distributed timers but instead utilize the structure of the real system considered.
Depending on the application, it may be possible for instance to pre-cut the jobs into smaller tasks of comparable size.
Finally, we would like to explore other variants of the queueing model, including
the representation of fork-join tasks using stochastic Petri networks \cite{petri},
the presence of negative customers \cite{G93,T13} and batch services \cite{BV91}
and the case of loss networks \cite{K91}.


\appendix

\renewcommand\theequation{A.\arabic{equation}}
\section*{Appendix}

We prove Theorem \ref{theo:stability} which states that the multi-server queue is stable if and only if
\begin{equation}
  \tag{\ref{eq:stability}}
  \forall A \subset \{1,\ldots,N\}, A \neq \emptyset, \quad
  \sum_{i \in A} \lambda_i < \mu(A).
\end{equation}

\paragraph{Necessary condition}

Assume that 
$\mu(A) \le  \sum_{i \in A} \lambda_i$ for some non-empty set $A \subset \{1,\ldots,N\}$.
For any $x \in \N^N \setminus \{0\}$ such that $A(x) \subset A$,
we also have $\mu(A(x)) \le \mu(A)$ since $\mu$ is non-decreasing,
so that
$$
\Phi(x)
= \frac1{\mu(A(x))} \sum_{i \in A(x)} \Phi(x-e_i)
\ge \frac1{\sum_{i \in A} \lambda_i} \sum_{i \in A(x)} \Phi(x-e_i),
$$
and by induction,
$$
\Phi(x)
\ge \frac{n!}{\prod_{i\in A} x_i!} \left( \frac1{\sum_{i \in A} \lambda_i} \right)^{n},
$$
where $n=\sum_{i\in A} x_i$.
Hence we obtain
\begin{align*}
  \sum_{x \in \N^N} \Phi(x) \prod_{i=1}^N {\lambda_i}^{x_i}
  &\ge \sum_{\substack{x \in \N^N: \\ A(x) \subset A}} \Phi(x) \prod_{i=1}^N {\lambda_i}^{x_i}
  = \sum_{n \ge 0} \sum_{\substack{x \in \N^N: \\ A(x) \subset A, \\ \sum_{i \in A} x_i = n}}
  \Phi(x) \prod_{i=1}^N {\lambda_i}^{x_i}, \\
  &\ge \sum_{n \ge 0}
  \sum_{\substack{x \in \N^N: \\ A(x) \subset A, \\ \sum_{i \in A} x_i = n}}
  \frac{n!}{\prod_{i\in A} x_i!} \prod_{i \in A} \left( \frac{\lambda_i}{\sum_{i \in A} \lambda_i} \right)^{x_i}, \\
  &= +\infty.
\end{align*}

\paragraph{Sufficient condition}

We first prove the following lemma.

\begin{lemma}
  \label{lem:min}
  Let  $\Psi$ be such that $\Psi(0) = 1$ and for all $x \neq 0$,
  \begin{equation}
    \label{eq:capa}
    \sum_{i \in A(x)} \frac{\Psi(x-e_i)}{\Psi(x)} \le \mu(A(x)).
  \end{equation}
  Then  $\Phi(x) \le \Psi(x)$ for all $x \in \N^N$.
\end{lemma}

\bp
The proof is by induction on $n = \sum_{i=1}^N x_i$.
The condition is true for $n = 0$ since $\Phi(0) = \Psi(0) = 1$.
Now let $n \ge 1$ and assume that $\Phi(x) \le \Psi(x)$ for all $x \in \N^N$ with $\sum_{i=1}^N x_i \le n-1$.
For each $x \in \N^N$ with $\sum_{i=1}^N x_i = n$, we obtain
$$
\Psi(x)
\ge \frac{\sum_{i \in A(x)} \Psi(x-e_i)}{\mu(A(x))}
\ge \frac{\sum_{i \in A(x)} \Phi(x-e_i)}{\mu(A(x))}
= \Phi(x),
$$
where the first inequality holds because $\Psi$ satisfies \eqref{eq:capa}
and the second holds by the induction assumption.
\ep \\

Assume that the stability condition \eqref{eq:stability} is satisfied.
The proof consists in choosing a function $\Psi$ that satisfies the assumptions of Lemma \ref{lem:min} and such that
$$
\sum_{x \in \N^N} \Psi(x) \prod_{i=1}^N {\lambda_i}^{x_i} < +\infty.
$$
In view of \eqref{eq:stability}, there exists $\eta \in \R_+^N$ such that
\begin{align*}
 \forall i = 1,\ldots,N, \quad  \lambda_i < \eta_i
  \quad  \text{and} \quad
 \forall A \subset \{1,\ldots,N\},\quad  \sum_{i \in A} \eta_i < \mu(A).
\end{align*}
We can choose for instance
$$
\forall i = 1,\ldots,N, \quad
\eta_i
= \lambda_i
+ \frac12 \min_{A \subset \{1,\ldots,N\}: i \in A}
\left( \frac{\mu(A) - \sum_{j \in A} \lambda_j}{|A|} \right).
$$
Now consider the balance function $\Psi$ defined by
$$
\forall x \in \N^N, \quad
\Psi(x) = \prod_{i=1}^N \frac1{{\eta_i}^{x_i}}.
$$
We have $\Psi(0) = 1$ and, for each $x \in \N^N \setminus \{0\}$,
\begin{align*}
  \sum_{i \in A(x)} \frac{\Psi(x-e_i)}{\Psi(x)}
  &= \sum_{i \in A(x)} \eta_i
  < \mu(A).
\end{align*}
We can thus apply Lemma \ref{lem:min} to $\Psi$
and we deduce that $\Phi(x) \le \Psi(x)$ for all $x \in \N^N$.
It follows that
\begin{align*}
  \sum_{x \in \N^N} \Phi(x) \prod_{i=1}^N {\lambda_i}^{x_i}
  &\le \sum_{x \in \N^N} \Psi(x) \prod_{i=1}^N {\lambda_i}^{x_i}
  = \sum_{x \in \N^N} \prod_{i=1}^N \left( \frac{\lambda_i}{\eta_i} \right)^{x_i}
  < +\infty,
\end{align*}
which concludes the proof.

\end{document}